\documentclass[english,aip,jmp,preprint,preprintnumbers]{revtex4-1}
\usepackage[T1]{fontenc}
\usepackage[latin9]{inputenc}
\usepackage{amsthm}
\usepackage{amsmath}
\usepackage{amssymb}
\usepackage{esint}

\makeatletter
 
 \@ifundefined{textcolor}{}
 {%
   \definecolor{BLACK}{gray}{0}
   \definecolor{WHITE}{gray}{1}
   \definecolor{RED}{rgb}{1,0,0}
   \definecolor{GREEN}{rgb}{0,1,0}
   \definecolor{BLUE}{rgb}{0,0,1}
   \definecolor{CYAN}{cmyk}{1,0,0,0}
   \definecolor{MAGENTA}{cmyk}{0,1,0,0}
   \definecolor{YELLOW}{cmyk}{0,0,1,0}
 }
  \theoremstyle{definition}
  \newtheorem{defn}{\protect\definitionname}
  \theoremstyle{plain}
  \newtheorem{prop}{\protect\propositionname}
\theoremstyle{plain}
\newtheorem{thm}{\protect\theoremname}
\newenvironment{lyxlist}[1]
{\begin{list}{}
{\settowidth{\labelwidth}{#1}
 \setlength{\leftmargin}{\labelwidth}
 \addtolength{\leftmargin}{\labelsep}
 }}
{\end{list}}
  \theoremstyle{remark}
  \newtheorem*{rem*}{\protect\remarkname}

\makeatother

\usepackage{babel}
  \providecommand{\definitionname}{Definition}
  \providecommand{\propositionname}{Proposition}
  \providecommand{\remarkname}{Remark}
\providecommand{\theoremname}{Theorem}

\begin{document}

\preprint{}

\preprint{Draft \#1}

\title{On a ($\beta,q$)-generalized Fisher information and inequalities
involving $q$-Gaussian distributions\footnote{This is a preprint version that differs from the published version, J. Math. Phys., vol. 53, issue. 6, 063303, doi:10.1063/1.4726197, in minor revisions, pagination and typographics details.}}

\author{J.-F. Bercher}

\email[To whom correspondence should be addressed: ]{jf.bercher@esiee.fr}

\homepage[]{http://www.esiee.fr/~bercherj}

\selectlanguage{english}%

\affiliation{Laboratoire d'informatique Gaspard Monge, UMR 8049  ESIEE-Paris,
Université Paris-Est}

\begin{abstract}
In the present paper, we would like to draw attention to a possible
generalized Fisher information that fits well in the formalism of
nonextensive thermostatistics. This generalized Fisher information
is defined for densities on $\mathbb{R}^{n}.$ Just as the maximum
Rényi or Tsallis entropy subject to an elliptic moment constraint
is a generalized $q$-Gaussian, we show that the minimization of the
generalized Fisher information also leads a generalized $q$-Gaussian.
This yields a generalized Cramér-Rao inequality. In addition, we show
that the generalized Fisher information naturally pops up in a simple
inequality that links the generalized entropies, the generalized Fisher
information and an elliptic moment. Finally, we give an extended Stam
inequality. In this series of results, the extremal functions are
the generalized $q$-Gaussians. Thus, these results complement the
classical characterization of the generalized $q$-Gaussian and introduce
a generalized Fisher information as a new information measure in nonextensive thermostatistics. 
\end{abstract}



\maketitle

\section{Introduction\label{sec:Section-I}}

The Gaussian distribution has a central role
with respect to standard information measures. For
instance, it is well known that the Gaussian distribution maximizes the entropy over all
distributions with the same variance; see \citet[Lemma 5]{dembo_information_1991}.
Similarly, the \citet{stam_inequalities_1959} inequality shows that
the minimum of the Fisher information over all distributions with
a given entropy also occurs for the Gaussian distribution. Finally,
the Cramér-Rao inequality, e.g. \citet[Theorem 20]{dembo_information_1991},
shows that the minimum of the Fisher information over all distributions
with a given variance is attained for the Gaussian distribution. 

It is thus natural to inquire for similar results for extended entropies
and an associated generalized Fisher information. In the context of
the nonextensive thermostatistics introduced by \citet{tsallis_possible_1988,tsallis_introduction_2009},
the role devoted to the standard Gaussian distribution is extended
to generalized $q$-Gaussian distributions, see e.g. \citet{lutz_anomalous_2003},
\citet{schwaemmle_q-gaussians_2008}, \citet{vignat_isdetection_2009},
\citet{ohara_information_2010}. The generalized $q$-Gaussians are
also explicit extremal functions of Sobolev, log-Sobolev or Gagliardo\textendash{}Nirenberg
inequalities on $\mathbb{R}^{n},$ with $n\geq2$, see \citet{del_pino_best_2002,del_pino_optimal_2003}.
This family of distributions will precisely achieve the equality case
in all the information inequalities presented in the sequel.  
\begin{defn}
 Let $x$ be a random vector of $\mathbb{R}^{n}.$ For $\alpha\in(0,\infty),$
$\gamma$ a real positive parameter and $q>(n-\alpha)/n,$ the generalized
$q$-Gaussian with scale parameter $\gamma$ has the radially symmetric
probability density 
\begin{equation}
G_{\gamma}(x)=\begin{cases}
\frac{1}{Z(\gamma)}\left(1-\left(q-1\right)\gamma|x|^{\alpha}\right)_{+}^{\frac{1}{q-1}} & \text{for }q\not=1\\
\frac{1}{Z(\gamma)}\exp\left(-\gamma|x|^{\alpha}\right) & \text{if }q=1
\end{cases}\text{ }\label{eq:qgauss_general}
\end{equation}
where $|x|$ denotes the standard euclidean norm, where we use the notation $\left(x\right)_{+}=\mbox{max}\left\{ x,0\right\} $,
and where $Z(\gamma)$ is the partition function such that $G_{\gamma}(x)$
integrates to one. Its expression is given in (\ref{eq:GenPartitionFunction}).
\end{defn}
For $q>1$, the density has a compact support, while for $q\leq1$
it is defined on the whole $\mathbb{R}^{n}$ and behaves as a power
distribution for $|x|\rightarrow\infty.$ Notice that the name generalized
Gaussian is sometimes restricted to the case $q=1$ above. In the
following, we will simply note $G,$ rather than $G_{1},$ the generalized
$q$-Gaussian obtained with $\gamma=1.$ The general expressions of
the main information measures attached to the generalized Gaussians
are derived in Appendix \ref{sec:Information-measures-of}. 

Let us recall that if $x$ is a vector of $\mathbb{R}^{n}$ with probability
density $f(x)$ with respect to the Lebesgue measure, then the information
generating function introduced by \citet{golomb_information_1966}
is defined by 
\begin{equation}
M_{q}[f]=\int f(x)^{q}\mathrm{d}x,\label{eq:DefInfoGeneratingFunction}
\end{equation}
for $q\geq0$. Associated with this information generating function,
the Rényi entropy is given by 
\begin{equation}
H_{q}[f]=\frac{1}{1-q}\log M_{q}[f]\label{eq:DefRenyi}
\end{equation}
while the Tsallis entropy is given by 
\begin{equation}
S_{q}[f]=\frac{1}{1-q}\left(1-M_{q}[f]\right).\label{eq:DefTsallis}
\end{equation}
Finally, we will also call ``entropy power'' of order $q$, or $q$-entropy power, the quantity
\begin{equation}
N_{q}[f]=M_{q}[f]^{\frac{1}{1-q}}=\exp\left(H_{q}[f]\right)=\left(\int f^{q}\text{d}x\right)^{\frac{1}{1-q}},\label{eq:defEntropyPower}
\end{equation}
for $q\neq1.$ For $q=1,$ we let $N_{q}[f]=\exp\left(H_{q}[f]\right).$
Note that the entropy power is usually defined as the square of the
quantity above, with an additional factor. We will also use the definition
of the elliptic moment of order $\alpha$, given by
\begin{equation}
m_{\alpha}[f]=\int|x|^{\alpha}f(x)\mathrm{d}x, \label{eq:DefEllipticMoment}
\end{equation}
where $|x|$ denotes the euclidean norm of $x$. 

In the context of nonextensive thermostatistics, the study of the
role of Fisher information, or the definition of an extension has
attracted some efforts.  Several authors have introduced generalized
versions of the Fisher information and derived the corresponding Cramér-Rao
bounds. Among these contributions we note the series of papers \citet{chimento_naudts-like_2000,casas_fisher_2002,pennini_rnyi_1998,pennini_semiclassical_2007}
where the proposed extended Fisher information involves a power $f^{q}$
of the original distribution $f$, or a normalized version which is
the escort distribution. We shall point out here the early contribution
of P. Hammad \citep{hammad_mesure_1978} that essentially introduced
the previous escort-Fisher information. A modified version, that also
involves escort distributions, is considered by \citet{naudts_generalised_2008,naudts_q-exponential_2009}.
A recent contribution in the nonextensive literature and in this journal
is the proposal of \citet{furuichi_maximum_2009,furuichi_generalized_2010},
where the logarithm is replaced by a deformed one and a Cramér-Rao
type inequality is given. 

In the present paper, we present an extension of the Fisher information
measure, the $(\beta,q)$-Fisher information $I_{\beta,q}[f]$ that
depends on an entropic index $q$ and on a parameter $\beta.$ This
information has been introduced by \citet{lutwak_Cramer_2005}. It
reduces to the standard Fisher information in the case $q=1,$$\beta=2.$
We look for the probability density function with minimum generalized
Fisher information, among all probability density functions with a
given elliptic moment, let us say $m_{\alpha}[f]=m$:
\begin{equation}
I_{\beta,q}(m)=\inf_{f}\{I_{\beta,q}[f]:\,\,\, f(x)\geq0,\,\,\,\int f(x)\mathrm{d}x\mbox{=1\,\,\,\ and\,\,\,}m_{\alpha}[f]=m\}.\label{eq:pb_new1}
\end{equation}
It will turn out that the extremal distribution is precisely the generalized
$q$-Gaussian with the prescribed elliptic moment. Actually, there
is a more precise statement of the result in the form of a Cramér-Rao
inequality that links the Fisher information and the elliptic moment
$m_{\alpha}[f]$. It is shown in section \ref{sec:The-generalized-Stam}
that for $\alpha$ and $\beta$ Hölder conjugates of each other, with
$\alpha\geq1$, then 
\begin{equation}
I_{\beta,q}[f]^{\frac{1}{\beta\lambda}}m_{\alpha}[f]^{\frac{1}{\alpha}}\geq I_{\beta,q}[G]^{\frac{1}{\beta\lambda}}m_{\alpha}[G]^{\frac{1}{\alpha}}\label{eq:genCR}
\end{equation}
with $\lambda=n(q-1)+1,$ and where the equality is obtained for any
generalized $q$-Gaussian distribution $f=G_{\gamma}$. This inequality
means in particular that the minimum of the generalized Fisher information
subject to an elliptic moment constraint is reached by a generalized
Gaussian. This complements the classical characterization of the generalized
$q$-Gaussian that are found as the maximum Rényi or Tsallis entropy
distributions with a given moment. Incidentally, we will also find,
from a general sharp Gagliardo-Nirenberg inequality, that the generalized
entropy power and Fisher information are linked by an inequality that
generalizes the Stam inequality
\begin{equation}
N_{q}[f]I_{\beta,q}[f]^{\frac{n}{\beta\lambda}}\geq N_{q}[G]I_{\beta,q}[G]^{\frac{n}{\beta\lambda}},\label{eq:genStam}
\end{equation}
again with $\lambda=n(q-1)+1$. 

Actually, the inequalities (\ref{eq:genCR}) and (\ref{eq:genStam})
have been presented in the monodimensional case in the remarkable
paper by \citet{lutwak_Cramer_2005} whose reading has deeply influenced
this work. Beyond the highlighting on the result in the nonextensive
statistics context, a contribution of the present paper is the extension
of the definitions and results in the multidimensional case. Furthermore,
we show in section \ref{sec:A-Fisher-Moment-entropy-inequali}, that
the expression of the generalized Fisher information naturally pops
up from a simple application of Hölder's inequality, which leads to
an inequality that links the generalized Fisher information, the Rényi
entropy power and the elliptic moment. We show that this inequality,
which is saturated by the generalized $q$-Gaussian, immediately suggests
the Cramér-Rao inequality (\ref{eq:genCR}). This Cramér-Rao inequality
is derived in section \ref{sec:The-generalized-Stam} with the help
of the generalized Stam inequality (\ref{eq:genStam}).

\section{\label{sec:On-the-generalized}On the generalized $(\beta,q)$-Fisher
information}

We begin by the definition of the generalized $(\beta,q)$-Fisher
information and its main consequences. 
\begin{defn}
Let $f(x)$ be a probability density function defined over a subset
$\Omega$ of $\mathbb{R}^{n}$. Let $|x|$ denote the euclidean norm
of $x$ and $\nabla f$ the gradient of function $f(x)$ with respect to $x$. If $f(x)$ is continuously
differentiable over $\Omega,$ then for $q\geq0$, $\beta>1$, the
generalized $(\beta,q)$-Fisher information is defined by
\end{defn}
\begin{equation}
I_{\beta,q}[f]=\int_{\Omega}f(x)^{\beta(q-1)+1}\left(\frac{|\nabla f(x)|}{f(x)}\right)^{\beta}\mathrm{d}x.\label{eq:GenFisherDefinition}
\end{equation}

The standard Fisher information ($q=1,$ $\beta=2)$, see for instance
\citet{cohen_fisher_1968,huber_robust_2009,plastino_fisher_2005}
is a convex functional of the probability density $f.$ It was shown
by \citet{boekee_extension_1977} that this is also true for the generalized
Fisher information with $q=1$ and any $\beta>1.$ It is also easy
to check that this is still true in the case $\beta(q-1)+1=\beta$.
Unfortunately, the general characterization of the convexity or pseudo-convexity
properties of the generalized Fisher information remains an open problem.

A radially symmetric function $f(x)$ only depends on the euclidean distance  $r=|x|$ of the argument from the origin. 
Hence, in the case of radially symmetric function defined on a symmetric
domain $\Omega,$ we have $f(x)=f_{r}(|x|)=f_{r}(r)$, where $f_{r}$
is an univariate probability density defined on a subset $\Omega_{r}$
of $\mathbb{R}^{+}.$ With these notations, we simply have $|\nabla f(x)|=\left|\frac{\mathrm{d}f_{r}(r)}{\mathrm{d}r}\right|$.
In polar coordinates, with $r=|x|,$ $\text{d}x=r^{n-1}\text{d}r\text{d}u$
and $\int\text{d}u=n\,\omega_{n},$ where $\omega_{n}$ is the volume
of the $n$-dimensional unit ball, the expression of the generalized
Fisher information becomes
\[
I_{\beta,q}[f]=n\,\omega_{n}\,\int_{\Omega_{r}}r^{n-1}f_{r}(r)^{\beta(q-1)+1}\left|\frac{f_{r}'(r)}{f_{r}(r)}\right|^{\beta}\mathrm{d}r.
\]

In the following, it will also be convenient to use the simple transformation
$f(x)=u(x)^{k}$, with $k=\beta/\left(\beta(q-1)+1\right)$. Indeed,
with this notation and taking into account that $|\nabla f|=|k|u^{k-1}|\nabla u|,$
the generalized Fisher information reduces to
\begin{equation}
I_{\beta,q}[f]=|k|^{\beta}\int_{\Omega}|\nabla u(x)|^{\beta}\mathrm{d}x,\label{eq:GenFisherDefinition_reduced}
\end{equation}
which corresponds to the $\beta$-Dirichlet energy of $u(x)$. In
the case $k=2,$ the function $u(x)$ is analog to a wave function
in quantum mechanics.

The variational problem (\ref{eq:pb_new1}) can be restated as follows:
\begin{equation}
I_{\beta,q}(m)=\inf_{u}\{\int_{\Omega}|\nabla u(x)|^{\beta}\mathrm{d}x:\,\,\,\int_{\Omega}u(x)^{k}\mathrm{d}x=1\mbox{\,\,\,\ and\,\,\,}m=\int_{\Omega}|x|^{\alpha}u(x)^{k}dx\}.\label{eq:pbnew1_bis}
\end{equation}
The Lagrangian associated with the problem (\ref{eq:pbnew1_bis})
is 
\begin{equation}
L(u;a,b)=\int_{\Omega}|\nabla u(x)|^{\beta}\mathrm{d}x+a\int_{\Omega}u(x)^{k}\mathrm{d}x+b\int_{\Omega}|x|^{\alpha}u(x)^{k}dx
\end{equation}
where $a$ and $b$ are the Lagrange parameters associated with the
two constraints. It is remarkable to note that the same kind of formulation
has been introduced in the book by \citet{frieden_sciencefisher_2004},
and solutions derived in the cases $k=2$ and $k=1$, with the mention
of the physical interest of the latter case. The Euler-Lagrange 
equation, see \citet{brunt_calculus_2004}, corresponding to the variational problem (\ref{eq:pbnew1_bis}) is 
\begin{equation}
\mathrm{div}\left(|\nabla u(x)|^{\beta-2}\nabla u(x)\right)-\frac{k}{\beta}\,\left(a+b|x|^{\alpha}\right)\, u(x)^{k-1}=0,\label{eq:LaplaceEquation}
\end{equation}
It has the form of a $p$-Laplace equation as in \citet{lindqvist_notes_2006}, here with $p=\beta$, 
\[
\triangle_{\beta}u(x)-\frac{k}{\beta}\,\left(a+b|x|^{\alpha}\right)\, u(x)^{k-1}=0,
\]
where $\triangle_{\beta}$ denotes the $\beta$-Laplacian operator.
Observe that in the special case $k=2,$ (\ref{eq:LaplaceEquation})
is analogous to the Schrödinger equation for the quantum oscillator
in a relativistic setting (Klein-Gordon equation). In such case, it
is known that the ground state is a normal distribution. In the radial
case, with $r=|x|$ and $u(x)=u_{r}(r),$ the nonlinear second order
differential equation reduces to 
\begin{equation}
\left(r^{n-1}u_{r}'(r)^{\beta-1}\right)'-\frac{k}{\beta}\, r^{n-1}\left(a+br^{\alpha}\right)\, u_{r}(r)^{k-1}=0.\label{eq:emden-fowler-sca}
\end{equation}
In our context, it is thus natural to suspect that the generalized
$q$-Gaussian could be a solution. Actually, it is easy to check that
if $G_{\gamma}$ is the generalized $q$-Gaussian (\ref{eq:qgauss_general}),
then $u(x)=G_{\gamma}(x)^{\frac{1}{k}}$ is indeed a solution of this
equation, with the following relations between the parameters:

\begin{equation}
a=-An,\,\, b=A(1+n(q-1))\gamma\text{ \,\ where\,\ }A=\left(\frac{\beta}{k}\right)^{\beta}\left(\frac{\gamma}{\beta-1}\right)^{\beta-1}Z(\gamma)^{k-\beta}.
\end{equation}

It will be shown in the following that the generalized Gaussian is
not only a possible solution, but the actual optimum solution of the
problem. 

Let us finally note that in the case $\beta=2,$ the nonlinear differential
equation (\ref{eq:emden-fowler-sca}) is an instance of the generalized
Emden-Fowler equation $u''(x)+h(x)\, u(x)^{\gamma}=0$, where $h(x)$
is a given function. This kind of equations arise in studies of gaseous
dynamics in astrophysics, in certain problems in fluid mechanics and
pseudoplastic flow, see \citet{nachman_nonlinear_1980}, 
as well as in some reaction-diffusion processes, 
as mentioned by \citet{covei_existence_2010}. 

In the radial case, the minimum Fisher information can be written
in terms of the constraints and of the Lagrange parameters associated
with these constraints. 
\begin{prop}
\label{pro:minFishLag}Let $\Omega$ be the $n$-dimensional ball
of radius $R$, possibly infinite, centered on the origin. Among all
radial densities $f=u^{k}$ defined on $\Omega$ and such that $f_{r}(R)=0,$
the minimum Fisher information in (\ref{eq:pbnew1_bis}) can be expressed
in terms of the constraints and the Lagrange multipliers as 
\begin{equation}
\frac{1}{|k|^{\beta}}I_{\beta,q}(m)=-\frac{k}{\beta}\left(a+bm\right).\label{eq:minFish}
\end{equation}
\end{prop}
\begin{proof}
By integration by parts, 
\begin{alignat}{1}
n\,\omega_{n}\,\int_{0}^{R}r^{n-1}\left(u_{r}'(r)^{\beta}\right)\mathrm{d}r & =n\,\omega_{n}\,\int_{0}^{R}r^{n-1}\left(u_{r}'(r)^{\beta-1}\right)u_{r}'(r)\mathrm{d}x\label{eq:toto}\\
= & n\,\omega_{n}\,\left[\left(r^{n-1}u_{r}'(x)^{\beta-1}\right)\, u_{r}(x)\right]_{0}^{R}-n\,\omega_{n}\,\int_{0}^{R}u_{r}(r)\left(r^{n-1}u_{r}'(r)^{\beta-1}\right)'\mathrm{d}x
\end{alignat}
Using the boundary condition and the differential equation (\ref{eq:emden-fowler-sca}),
we then obtain
\[
n\,\omega_{n}\,\int_{0}^{R}r^{n-1}\left(u_{r}'(r)^{\beta}\right)\mathrm{d}r=-n\,\omega_{n}\,\frac{k}{\beta}\int_{0}^{R}r^{n-1}\left(a+br^{\alpha}\right)\, u_{r}(r)^{k}\mathrm{d}x.
\]
Taking into account the values of the constraints, this reduces to
(\ref{eq:minFish}) . 
\end{proof}
It can be noted that the Lagrange parameters are actually (complicated)
functions of the constraints, so that the right hand side of (\ref{eq:minFish})
is not a simple affine function in $m$. 

%

\section{\label{sec:A-Fisher-Moment-entropy-inequali}A Fisher-Moment-entropy
inequality and the extended Cramér-Rao inequality}

We show here that the generalized Fisher information naturally pops
up in a simple inequality that links this generalized Fisher information,
the information generating function and the elliptic moment of order
$\alpha.$ Then, we show that this inequality suggests the general
Cramér-Rao inequality (\ref{eq:genCR}). 
\begin{thm}
\label{theo_Fisher_moment_entropy}Let $f(x)=f_{r}(|x|)$ a radially
symmetric probability density on the $n$-dimensional ball $\Omega$
of radius $R$, possibly infinite, centered on the origin. Assume
that the density is absolutely continuous and such that $\lim_{r\rightarrow R}r^{n}\, f_{r}(r)^{q}=0$.
Let also $\alpha\geq1$ and $\beta$ be its Hölder conjugate. Then,
for $q>n/(n+\alpha)$ and provided that the involved information measures
are finite, we have 
\begin{equation}
I_{\beta,q}[f]^{\frac{1}{\beta}}\, m_{\alpha}[f]^{\frac{1}{\alpha}}\geq\frac{n}{q}M_{q}[f]\label{eq:Ineg_Fisher_Renyi_moment}
\end{equation}
with equality if and only if $f$ is a generalized Gaussian $f=G_{\gamma}$
for any $\gamma>0.$ \end{thm}
\begin{proof}
Let us consider the information generating function $M_{q}[f].$ Since
the density is assumed radially symmetric, the use of polar coordinates
reduces the computation to the evaluation of an univariate integral.
By integration by parts, we obtain
\begin{align}
M_{q}[f] & =\!\int_{\Omega}\! f(x)^{q}\text{d}x=n\omega_{n}\int_{0}^{R}\! r^{n-1}f_{r}(r)^{q}\text{d}r=n\omega_{n}\left[\frac{r^{n}}{n}\, f_{r}(r)^{q}\right]_{0}^{R}-n\omega_{n}q\!\int_{0}^{R}\!\frac{r^{n}}{n}\,{f_{r}'}(r)\, f_r(r)^{q-1}\text{d}r\\
 & =n\omega_{n}\frac{q}{n}\!\int_{0}^{R}\!\left[r\right]\,\left[-\frac{{f}_{r}'(r)}{f_{r}(r)}f_{r}(r)^{\left(q-1\right)}\right]\, r^{n-1}\, f_{r}(r)\text{d}r \label{eq:remainingintegral}
\end{align}
where we used the fact that $\lim_{r\rightarrow R}r^{n}\, f_{r}(r)^{q}=0$.
We can then use the fact that $\left|\int g(x) \text{d}x \right| \leq \int \left|g(x)\right| \text{d}x$ with equality iff $g(x)\geq 0$, and the Hölder inequality which states that 
\[ 
\left(\int |u(x)|^\alpha w(x) \text{d}x\right)^\frac{1}{\alpha}  \left(\int |v(x)|^\beta w(x) \text{d}x\right)^\frac{1}{\beta} \geq \int |u(x)v(x)| w(x) \text{d}x
\]
with $\alpha^{-1}+\beta^{-1}=1$, $\alpha>1$ and  where $w(x)\geq0$ is a weight function. The equality case then occurs iff $|u(x)|^\alpha = k |v(x)|^\beta$, with $k>0$. 

Let us now apply these inequalities to the remaining integral in (\ref{eq:remainingintegral}), with $g(r)=u(r)v(r)w(r)$, $u(r)=r$, $v(r)=-\frac{{f}_{r}'(r)}{f_{r}(r)}f_{r}(r)^{\left(q-1\right)}$,  and $w(r)=r^{n-1}\, f_{r}(r)$.  Immediately, we obtain
\begin{equation}
\left(\! n\omega_{n}\int_{0}^{R}\!\left|\frac{{f_{r}'}(r)}{f_{r}(r)}\right|^{\beta}f_{r}(r)^{\beta(q-1)}\, r^{n-1}\, f_{r}(r)\text{d}r\right)^{\frac{1}{\beta}} \left(n\omega_{n}\int_{0}^{R}\! r^{\alpha}\, r^{n-1}\, f_{r}(r)\text{d}r\right)^{\frac{1}{\alpha}} \geq\!\frac{n}{q}M_{q}[f]\label{eq:HolderIneq1-1}
\end{equation}
which is inequality (\ref{eq:Ineg_Fisher_Renyi_moment}) in polar
coordinates. 

The conditions for equality give here ${f_r'}(r)<0$  on the one hand, and $r^{\alpha}=k\left|\frac{{f}_{r}'(r)}{f_{r}(r)}\right|^{\beta}f_{r}(r)^{\beta(q-1)}$, $k>0,$ on the other hand. Finally,
using the fact that $\alpha/\beta=\alpha-1$, the first-order nonlinear
differential equation reduces to 
\begin{equation}
r^{\alpha-1}f_{r}(r)^{2-q}=-K\,{f_{r}'}(r).\label{eq:DiffEquationGenGauss}
\end{equation}
Then, it is easy to check that the unique normalized solution of the
differential equation (\ref{eq:DiffEquationGenGauss}) is nothing
but the generalized Gaussian (\ref{eq:qgauss_general}).
\end{proof}
We shall mention that in the monodimensional case, the radial symmetry
hypothesis is not necessary, as the integration can be achieved on
the real line, and the previous result holds for all probability densities.
In the case $q=1$, the information generating function in the right
side of (\ref{eq:Ineg_Fisher_Renyi_moment}) equals to 1, so that
the inequality reduces to $I_{\beta,1}[f]^{\frac{1}{\beta}}\, m_{\alpha}[f]^{\frac{1}{\alpha}}\geq1$
which is a Cramér-Rao inequality that has been exhibited by \citet{boekee_extension_1977}.

As a direct consequence, we obtain a possible derivation of the more
general Cramér-Rao inequality (\ref{eq:genCR}). The idea is to lower
bound the right side of (\ref{eq:Ineg_Fisher_Renyi_moment}). Let
us first consider the case $q>1.$ In this case, $M_{q}[f]$ is a
convex functional, and therefore has a single minimizer among all
densities with a given moment. It is well known that this minimizer
is a generalized $q$-Gaussian $G_{\theta}(x)$ with the same moment:
\[
M_{q}[f]\geq\inf_{p/m_{\alpha,1}[p]=m_{\alpha,1}[f]}M_{q}[p]=M_{q}[G_{\theta}].
\]
Therefore, in the case $q>1$ the inequality (\ref{eq:Ineg_Fisher_Renyi_moment})
also yields 
\begin{equation}
I_{\beta,q}[f]^{\frac{1}{\beta}}\, m_{\alpha,1}[f]^{\frac{1}{\alpha}}\geq\frac{n}{q}M_{q}[G_{\theta}],\label{eq:petitCramer}
\end{equation}
with equality if and only if $f=G_{\theta}$. Of course, the right
hand side term can also be written $I_{\beta,q}[G_{\theta}]^{\frac{1}{\beta}}\, m_{\alpha,1}[G_{\theta}]^{\frac{1}{\alpha}},$
or $I_{\beta,q}[G_{\theta}]^{\frac{1}{\beta}}\, m_{\alpha,1}[G_{\theta}]^{\frac{\lambda}{\alpha}}\, m_{\alpha,1}[G_{\theta}]^{\frac{1-\lambda}{\alpha}},$
with $\lambda=\left(n\left(q-1\right)+1\right)$. Therefore, the inequality
(\ref{eq:petitCramer}) becomes
\[
I_{\beta,q}[f]^{\frac{1}{\beta}}\, m_{\alpha,1}[f]^{\frac{1}{\alpha}}\, m_{\alpha,1}[G_{\theta}]^{-\frac{1-\lambda}{\alpha}}\geq I_{\beta,q}[G_{\theta}]^{\frac{1}{\beta}}\, m_{\alpha,1}[G_{\theta}]^{\frac{\lambda}{\alpha}}.
\]
Since the optimum distribution $G_{\theta}$ is such that $m_{\alpha,1}[f]=m_{\alpha,1}[G_{\theta}],$
we get
\begin{equation}
I_{\beta,q}[f]^{\frac{1}{\beta}}\, m_{\alpha,1}[f]^{\frac{\lambda}{\alpha}}\geq I_{\beta,q}[G_{\theta}]^{\frac{1}{\beta}}\, m_{\alpha,1}[G_{\theta}]^{\frac{\lambda}{\alpha}}.\label{eq:ptitCramerDeux}
\end{equation}
By the scaling properties (\ref{eq:ScaleEqualities}), the right hand
side does not depend on $\theta,$ so that the bound is attained for
any generalized Gaussian distribution. Finally, for $q\geq1,$ we
have $\lambda>0$ and (\ref{eq:ptitCramerDeux}) gives the Cramér-Rao
inequality (\ref{eq:genCR}). 

A similar approach holds in the case $q<1$ where we use the fact
that the minimum of Fisher information among distributions with a
given generalized Fisher information is attained by the generalized
$q$-Gaussian distribution with the same Fisher information. This
statement is a consequence of a generalized Stam inequality which
will be explicited in Proposition\,\ref{Prop_GeneralizedStam}. Hence,
we have
\[
M_{q}[f]\geq\inf_{p/I_{\beta,q}[p]=I_{\beta,q}[f]}\, M_{q}[p]=M_{q}[G_{\theta}]
\]
 and we can apply the same approach as above to obtain

\[
I_{\beta,q}[f]^{\frac{1}{\beta}}\, m_{\alpha,1}[f]^{\frac{1}{\alpha}}\, I_{\beta,q}[G_{\theta}]^{\frac{1}{\beta\lambda}(1-\lambda)}\,\geq I_{\beta,q}[G_{\theta}]^{\frac{1}{\beta\lambda}}\, m_{\alpha,1}[G_{\theta}]^{\frac{1}{\alpha}}.
\]
 Then, since $\theta$ is such that $I_{\beta,q}[f]=I_{\beta,q}[G_{\theta}]$
and since the right hand side does not depend on $\theta,$ we see
that the generalized Cramér-Rao (\ref{eq:genCR}) inequality is also
true for $q<1.$ This is summarized in the following Proposition. 
\begin{prop}
\label{prop_cramer-Rao-inequality}{[}Cramér-Rao inequality for radially
symmetric densities{]} If the conditions in Theorem \ref{theo_Fisher_moment_entropy}
hold, then we also have 
\begin{equation}
I_{\beta,q}[f]^{\frac{1}{\beta\lambda}}m_{\alpha}[f]^{\frac{1}{\alpha}}\geq I_{\beta,q}[G]^{\frac{1}{\beta\lambda}}m_{\alpha}[G]^{\frac{1}{\alpha}},\label{eq:CR_generalized-1}
\end{equation}
with $\lambda=n(q-1)+1$, and where the equality holds if and only
if $f$ is a generalized Gaussian $f=G_{\gamma}$. 
\end{prop}
Let us recall that this generalized Cramér-Rao inequality has been
exhibited by \citet{lutwak_Cramer_2005} in the monodimensional case.
The proposition above is restricted to radially symmetric densities
since the initial Fisher-Moment-entropy inequality in Theorem\,\ref{theo_Fisher_moment_entropy}
rests on this hypothesis. However, symmetrization arguments show that
the Cramér-Rao inequality (\ref{eq:CR_generalized-1}) should hold
more generally. Let $v$ be a function defined on a subset $\Omega$
of $\mathbb{R}^{n}$ and note $v^{*}$ the symmetric decreasing rearrangement
(Schwarz symmetrization) of $v$, and $v_{*}$ its symmetric increasing
rearrangement. The rearrangement inequality of \citet{hardy_inequalities_1988},
see also \citet{kesavan_symmetrization_2006} or \citet{r.d._benguria_isoperimetric_2008},
states that $\int_{\Omega}fg\geq\int_{\Omega}f^{*}g_{*}$. This can
be applied directly to the elliptic moment $m_{\alpha}[f]$ with $g(x)=|x|^{\alpha}.$
Since $|x|^{\alpha}$ is symmetric increasing, it is its own increasing
rearrangement and we obtain that $m_{\alpha}[f]\geq m_{\alpha}[f^{*}],$
which seems indeed a natural inequality. In addition, the famous Pólya\textendash{}Szeg\H{o}
inequality states that $\int_{\Omega}|\nabla u|^{\beta}\geq\int_{\Omega}|\nabla u^{*}|^{\beta}$.
Therefore, with $f=u^{k}$ and using the fact that $f^{*}=(u^{*})^{k}$,
we obtain that $I_{\beta,q}[f]\geq I_{\beta,q}[f^{*}].$ Hence, we
obtain that for any $f$ we always have 
\[
I_{\beta,q}[f]^{\frac{1}{\beta\lambda}}m_{\alpha}[f]^{\frac{1}{\alpha}}\geq I_{\beta,q}[f^{*}]^{\frac{1}{\beta\lambda}}m_{\alpha}[f^{*}]^{\frac{1}{\alpha}},
\]
which means that the minimizer of the left side is necessarily radially
symmetric and decreasing. As a result, the extremal function can be
sought in the subset of radially symmetric probability densities.
Unfortunately, the definition of the increasing rearrangement requires
that $\Omega$ is a compact subset of $\mathbb{R}^{n}$, so that the
argument fails for densities with non compact support (which corresponds
to the case $q<1$ of the extremal function). However, this still
shows that the Cramér-Rao inequality (\ref{eq:CR_generalized-1})
holds for any probability density in the $q>1$ case, and suggests
to investigate further the general case.

\section{\label{sec:The-generalized-Stam}The generalized Stam and Cramér-Rao
inequalities}

Proposition \ref{prop_cramer-Rao-inequality} relies on a result on
the minimization of the Rényi entropy power and on a result on the
minimization of the generalized Fisher information.  The fact that
the generalized Gaussian maximizes the Rényi entropy among all distributions
with a given elliptic moment is well known and can be derived by standard
calculus of variations. A precise statement under the form of a general
inequality is due to \citet{lutwak_moment-entropy_2007}. 
\begin{thm}
\label{theo_moment_entropy}{[}\citet{lutwak_moment-entropy_2007}
- Theorem 3{]} If $\alpha\in(0,\infty),$ $q>n/(n+\alpha),$ and f
is a probability density on random vectors of $\mathbb{R}^{n}$, with
$m_{\alpha}[f]<\infty$ its moment of order $\alpha$ and $N_{q}[f]<\infty$
its Rényi entropy power, then
\begin{equation}
\frac{m_{\alpha}[f]^{\frac{1}{\alpha}}}{N_{q}[f]^{\frac{1}{n}}}\geq\frac{m_{\alpha}[G]^{\frac{1}{\alpha}}}{N_{q}[G]^{\frac{1}{n}}},\label{eq:momententropy}
\end{equation}
 with equality if and only if $f$ is any generalized Gaussian in
(\ref{eq:qgauss_general}). 
\end{thm}
As a direct consequence we indeed have the maximum entropy characterization
of generalized Gaussians: among all probability densities with a given
moment, e.g. $m_{\alpha}[f]=m$, we have $N_{q}[f]\leq N_{q}[G_{\gamma}]$
with $\gamma$ such that $m_{\alpha}[G_{\gamma}]=m$. 

The second important result is the fact that the minimum of the generalized
Fisher information, among all probability densities with a given Rényi
entropy, is reached by a generalized Gaussian. This generalizes the
fact that the minimum of the standard Fisher information for probability
densities with a given (Shannon) entropy is attained by the standard
normal distribution, as shown by Stam's inequality. Actually, the
statement on generalized Fisher information relies on a generalized
version of Stam's inequality. We will use a remarkable general result
of \citet{agueh_sharp_2008} that gives the optimal functions in some
special cases of sharp Gagliardo-Nirenberg inequalities. For $n\geq2,$
Agueh's result recovers the results of \citet{del_pino_optimal_2003},
while the case $n=1$ can be obtained from a sharp Gagiardo-Nirenberg
inequality on the real line established by \citet{b._sz.-nagy_ueber_1941}.
It turns out that these optimal functions are the generalized Gaussians.
Let us first recall Agueh's result. 
\begin{thm}
{[}\citet{agueh_sharp_2008}- Theorem 2.1, corollary 3.4{]} For $n,p,r$
and $s$ such that
\begin{flalign}
n>p>1 & \mathrm{\,\,\, and\,\,} & \frac{np}{n-p}>s>r\geq1 & \mathrm{\,\,\, if\,\,\,} & n>1\label{eq:ConditionsGagliardo-Nirenberg}\\
p>1 & \mathrm{\,\,\, and\,\,} & \infty>s\geq r\geq1 & \mathrm{\,\,\, if\,\,\,} & n=1\nonumber 
\end{flalign}
and if $u(x)$ is a function defined on $\mathbb{R}^{n}$ such that
the involved norms are finite, then the following sharp Gagiardo-Nirenberg
inequality holds
\begin{equation}
K\,\left\Vert \left|\nabla u\right|\right\Vert _{p}^{\theta}\,\left\Vert u\right\Vert _{r}^{1-\theta}\geq\left\Vert u\right\Vert _{s}\label{eq:Gagliardo-Nirenberg}
\end{equation}
where $K$ is an optimum constant, and 
\[
\theta=\frac{np(s-r)}{s\left[np-r(n-p)\right]}.
\]
Let $p^{*}=p/(p-1)$ denote the Hölder conjugate of $p$. \end{thm}
\begin{lyxlist}{00.00.0000}
\item [{(a)}] If $r=1+s/p^{*},$ then the extremal functions have the form
$Cu(\sigma(x-\bar{x}))$ where $C,\sigma$ and $\bar{x}$ are arbitrary,
and 
\[
u(x)=\left(1+|x|^{p^{*}}\right)^{\frac{p}{p-s}}.
\]

\item [{(b)}] If $r=p^{*}\left(s-1\right),$ then the extremal functions
have the form $Cu(\sigma(x-\bar{x}))$ where $C,\sigma$ and $\bar{x}$
are arbitrary, and
\[
u(x)=\left(1-|x|^{p^{*}}\right)_{+}^{\frac{p-1}{p-s}}.
\]

\end{lyxlist}
This result enables to obtain a generalization of Stam's inequality
in $\mathbb{R}^{n},$ $n\geq1,$ involving the $q$-entropy power
and the generalized $(\beta,q)$-Fisher information. 
\begin{prop}
\label{Prop_GeneralizedStam}For $n\geq1,$ $\beta$ and $\alpha$
Hölder conjugates of each other, $\alpha>1,$ $q>\max\left\{ (n-1)/n,n/(n+\alpha)\right\} $
then for any probability density on $\mathbb{R}^{n}$, supposed continuously
differentiable, 
\begin{equation}
N_{q}[f]I_{\beta,q}[f]^{\frac{n}{\beta\lambda}}\geq N_{q}[G]I_{\beta,q}[G]^{\frac{n}{\beta\lambda}},\label{eq:GeneralizedStam}
\end{equation}
with $\lambda=n(q-1)+1,$ and with equality if and only if $f$ is
any generalized Gaussian $G_{\gamma}$ in (\ref{eq:qgauss_general}).\end{prop}
\begin{proof}
The result mainly follows from Agueh's theorem above. Take $\beta=p,$
$\alpha=p^{*}$ and let $u(x)=f(x)^{t}.$ 

Let us first consider case (a), and choose the exponent $t>0$ such
that $st=1.$ Observe that the condition $s>r$ gives $s>\beta.$
Finally, let us denote $q=rt=(1+s/\alpha)t$. In such conditions,
the sharp Gagliardo-Nirenberg inequality (\ref{eq:Gagliardo-Nirenberg})
becomes
\[
K\,\left(\int f(x)^{q}\mathrm{d}x\right)^{\frac{(1-\theta)t}{q}}\left(\int f(x)^{\beta t}\left(\frac{|\nabla f(x)|}{f(x)}\right)^{\beta}\mathrm{d}x\right)^{\frac{\theta}{\beta}}\geq1,
\]
since $f$ is a probability density. The relationship $q=(1+s/\alpha)t=t+1-1/\beta$
implies that $\beta t=\beta(q-1)+1$. Hence, we see that the simple
change of variable $u(x)=f(x)^{t}$ yields a relation that involves
both the information generating function and the generalized Fisher
information. Furthermore, we have $\beta/(\beta-s)=\beta t/(\beta t-1)=t/(q-1)<0$
so that the optimum function $f(x)=u(x)^{1/t}$ is nothing but the
generalized Gaussian. Finally, the direct computation of the exponents
leads to 
\[
\left(M_{q}[f]^{\frac{1}{1-q}}I_{\beta,q}[f]^{\frac{n}{\beta\left[n(q-1)+1\right]}}\right)^{a}\, K\geq1
\]
with $1>q>\max\left\{ (n-1)/n,n/(n+\alpha)\right\} $ and $a=t\frac{\left[n(\beta(q-1)+1)-(n-\beta)\right]}{n\left[\beta(q-1)+1\right]-q(n-\beta)}>0$,
and where we note that $M_{q}[f]^{\frac{1}{1-q}}=N_{q}[f]$. Taking
into account that the equality sign only holds for the generalized
Gaussians and that the corresponding expression is scale invariant
(as a consequence of (\ref{eq:ScaleEqualities})), we arrive at the
inequality (\ref{eq:GeneralizedStam}). 

The approach is similar in case (b) where $r=\alpha\left(s-1\right).$
Observe that the condition $s>r$ gives here $s<\beta.$ Take $st=q$,
with $t>0,$ and $rt=\alpha t\left(s-1\right)=1$. These two equalities
give $\beta t=\beta\left(q-1\right)+1$ which, since $\beta>s$, gives
$q>1$. Hence, we obtain that $f(x)=u(x)^{1/t}$ is the generalized
Gaussian with exponent $1/(q-1)$ and compact support. Finally, the
simplification of the Gagliardo-Nirenberg inequality shows that the
inequality (\ref{eq:GeneralizedStam}) holds. 

It remains to examine the case $q=1$ in (\ref{eq:GeneralizedStam}).
In this case, the result can be obtained as a direct consequence of
an Lp log-Sobolev inequality, which can be viewed as the limit case
of the Sharp Gagliardo-Nirenberg inequality when $p\downarrow s$,
cf. \citet{del_pino_optimal_2003}. For instance, exponentiating the
Lp logarithmic Sobolev inequality in \citet[eq. (1)]{gentil_general_2003}
\[
\int |f|^p\log |f|^p \text{d}x \leq \frac{n}{p}\log \left(K_p \int |\nabla f|^p \text{d}x\right)
\]
where $K_p$ is an optimal constant, and using the change of variable $f=|f|^{p}$ directly gives the generalized
Stam inequality (\ref{eq:GeneralizedStam}), with equality for the
generalized Gaussians with $q=1.$ \end{proof}
\begin{rem*}
Let us mention that the case $q=1$ has already be established as
a direct consequence of the Lp logarithmic Sobolev inequality in the
recent paper by \citet{kitsos_logarithmic_2009}. Let us also note
that another kind of inequality with a generalized Gaussian extremal
can also be derived from the standard Sobolev inequality (that is
the inequality (\ref{eq:Gagliardo-Nirenberg}) with $\theta=1).$ 

The monodimensional case of (\ref{eq:GeneralizedStam}) has be derived
in a very elegant way in the work by Lutwak, Yang and Zhang in \citet{lutwak_Cramer_2005}.
It can also be derived from \citet{b._sz.-nagy_ueber_1941}'s sharp
Gagliardo-Nirenberg inequality on $\mathbb{R}$. 
\end{rem*}
Finally, we obtain that the generalized Fisher information $I_{\beta,q}[f]$
and the elliptic moments of order $\alpha$ $m_{\alpha}[f]$ are indeed
linked by a Cramér-Rao inequality, in all generality. The related
inequality generalizes the standard Cramér-Rao inequality for the
location parameter. 
\begin{thm}
{[}Cramér-Rao inequality{]} For $n\geq1,$ $\beta$ and $\alpha$
Hölder conjugates of each other, $\alpha>1,$ $q>\max\left\{ (n-1)/n,n/(n+\alpha)\right\} $
then for any probability density $f$ on $\mathbb{R}^{n}$, supposed
continuously differentiable and such that the involved information
measures are finite, we have 
\begin{equation}
I_{\beta,q}[f]^{\frac{1}{\beta\lambda}}m_{\alpha}[f]^{\frac{1}{\alpha}}\geq I_{\beta,q}[G]^{\frac{1}{\beta\lambda}}m_{\alpha}[G]^{\frac{1}{\alpha}},\label{eq:CR_generalized}
\end{equation}
with $\lambda=n(q-1)+1$, and where the equality holds if and only
if $f$ is a generalized Gaussian $f=G_{\gamma}$. \end{thm}
\begin{proof}
The result is an immediate consequence of the moment-entropy inequality
(\ref{eq:momententropy}) and of the generalized Stam inequality (\ref{eq:GeneralizedStam}):
the simple term by term product of both inequalities directly gives
(\ref{eq:CR_generalized}), with the same equality condition as in
the two initial inequalities. 
\end{proof}

\section{Conclusions}

In this paper, we have presented a generalized Fisher information
measure that fits well in the nonextensive thermostatistics context.
Indeed, just as the maximization of the generalized Rényi or Tsallis
entropies subject to a moment constraint yields a $q$-Gaussian distribution,
the minimization of the generalized Fisher information subject to
the same constraint also leads to the very same $q$-Gaussian distribution.
A generalized Cramér-Rao inequality corresponds to this result. Furthermore
a generalized Stam inequality links the generalized entropies and
the generalized Fisher information, with a lower bound attained, again,
by $q$-Gaussian distributions. All these results hold for probability
densities defined on $\mathbb{R}^{n}.$ Hence, these results complement
the classical characterization of the generalized $q$-Gaussian and
introduce a generalized Fisher information as a new information measure
associated with Rényi or Tsallis entropies. While this work was under consideration, 
a similar extension of the generalized Fisher information to the multidimensional case 
has been proposed by \citet{lutwak_extension-fisher_2012} in the context of information 
theory.

Future work will include the study of further properties of the generalized
Fisher information, namely its convexity properties. Here, the generalized
Fisher information has been introduced as the information attached
to the distribution. In estimation theory, the Fisher information
is defined with respect to a general parameter and characterize the
information about this parameter, as well as the estimation performances,
as exemplified by the classical Cramér-Rao bound in estimation theory.
Hence, it would be of interest to look at general estimation problems
that could involve a similar generalized Fisher information. In nonextensive
thermostatistics, the notion of escort distributions is an important
ingredient related to the generalized entropies. Thus, it would also
be of interest to see how these escort distributions can be introduced
in the present setting. 

\section*{Acknowledgments}
The author thanks the anonymous referee for his valuable comments and suggestions that helped improve the presentation
of this article. Thanks are extended to Lodie Garbell for her friendly proofreading of the
manuscript.

\begin{turnpage}

\end{turnpage}

\begin{ruledtabular}
\end{ruledtabular}
\begin{turnpage}

\end{turnpage}
\appendix

\section{\label{sec:Information-measures-of}Information measures of generalized
Gaussians}

In order to get the expressions of the information measures associated
to the generalized Gaussian (\ref{eq:qgauss_general}), that in turn
give the explicit expressions of the bounds in (\ref{eq:CR_generalized-1}),
(\ref{eq:momententropy}) and (\ref{eq:GeneralizedStam}), (\ref{eq:CR_generalized}),
we use the following result:
\begin{prop}
Let 
\[
\mu_{p,\nu}=\int|x|^{p}\left(1-s\gamma|x|^{\alpha}\right)_{+}^{\frac{\nu}{s}}\text{d}x
\]
with $\alpha,\gamma>0.$ By direct calculation in polar coordinates,
one gets 
\begin{alignat}{1}
\mu_{p,\nu} & =\frac{2}{\alpha}\left(\gamma\right)^{-\frac{p+n}{\alpha}}n\,\omega_{n}\times\nonumber \\
 & \begin{cases}
(-s)^{-\frac{p+n}{\alpha}}B\left(\frac{p+n}{\alpha},-\frac{\nu}{s}-\frac{p+n}{\alpha}\right) & \text{for }-\frac{\nu\alpha}{\left(p+n\right)}<s<0\\
s^{-\frac{p+n}{\alpha}}B\left(\frac{p+n}{\alpha},\frac{\nu}{s}+1\right) & \text{for }s>0\\
\left(\nu\right)^{-\frac{p+n}{\alpha}}\Gamma\left(\frac{p+n}{\alpha}\right) & \text{if }s=0
\end{cases}\label{eq:general_moment_pnu}
\end{alignat}
 where $\omega_{n}$ is the volume of the $n$-dimensional ball.
\end{prop}
From this expression, we immediately identify that the partition function
is $Z(\gamma)=\mu_{0,1},$ with $s=q-1$, that is 
\begin{equation}
Z(\gamma)=\frac{2}{\alpha}\left(\gamma\right)^{-\frac{n}{\alpha}}n\,\omega_{n}\times\begin{cases}
(1-q)^{-\frac{n}{\alpha}}B\left(\frac{n}{\alpha},-\frac{1}{q-1}-\frac{n}{\alpha}\right) & \text{for }1-\frac{\alpha}{n}<q<1\\
(q-1)^{-\frac{n}{\alpha}}B\left(\frac{n}{\alpha},\frac{1}{q-1}+1\right) & \text{for }q>1\\
\Gamma\left(\frac{n}{\alpha}\right) & \text{if }q=1.
\end{cases}\label{eq:GenPartitionFunction}
\end{equation}
 Similarly, we obtain the information generating function
\begin{equation}
M_{q}[G_{\gamma}]=\int G_{\gamma}(x)^{q}\text{d}x=\frac{\int g_{\gamma}(x)^{q}\text{d}x}{\left(\int g_{\gamma}(x)\text{d}x\right)^{q}}=\frac{\mu_{0,q}}{\left(\mu_{0,1}\right)^{q}},\label{eq:InfoGeneratingFunction}
\end{equation}
with $s=q-1$. The information generating function, and thus the associated
Rényi and Tsallis entropies are finite for $q>n/(n+\alpha).$ 

Likewise, the moment of order $p$ is given by
\begin{equation}
m_{p}[G_{\gamma}]=\frac{\int|x|^{p}g_{\gamma}(x)\text{d}x}{\int g_{\gamma}(x)\text{d}x}=\frac{\mu_{p,1}}{\mu_{0,1}}.
\end{equation}
If $p=\alpha,$ by the properties of the Beta functions, the expressions
for the moment of order $p$ simplifies into
\begin{equation}
m_{\alpha}[G_{\gamma}]=\frac{\mu_{\alpha,1}}{\mu_{0,1}}=\frac{n}{\alpha}\frac{1}{\gamma(q-1)\left(\frac{1}{q-1}+\frac{n}{\alpha}+1  \right)}\,\text{for }q>n/(n+\alpha)\label{eq:generalized_nu_moment_p=00003D00003Dalpha}
\end{equation}

Let us now consider the generalized Fisher information $I_{\beta,q}[f]$
defined in (\ref{eq:GenFisherDefinition}). For a radially symmetric
function, that is $f(x)=f(|x|)=f(r),$ we simply have $|\nabla f(x)|=\frac{\mathrm{d}f(r)}{\mathrm{d}r},$
and in the case of the generalized Gaussian (\ref{eq:qgauss_general}),
we obtain that $\left|{G}_{\gamma}'/G_{\gamma}\right|=\left|{g}_{\gamma}'/g_{\gamma}\right|=\alpha\gamma|x|^{\alpha-1}\left(1-\gamma(q-1)|x|^{\alpha}\right)_{+}^{-1}$
for $q\neq1$, so that the generalized Fisher information has the
expression
\begin{alignat}{1}
I_{\beta,q}[G_{\gamma}] & =\frac{\left(\alpha\gamma\right)^{\beta}}{\left(\mu_{0,1}\right)^{\beta(q-1)+1}}\int|x|^{\alpha}\left(1-\gamma(q-1)|x|^{\alpha}\right)_{+}^{\frac{\beta(q-1)+1}{(q-1)}-\beta}\mathrm{d}x\label{eq:general_expression_for_FisherA}\\
 & =\left(\alpha\gamma\right)^{\beta}\frac{\mu_{\alpha,1}}{\left(\mu_{0,1}\right)^{\beta(q-1)+1}}.\label{eq:general_expression_for_Fisher}
\end{alignat}
We easily obtain that (\ref{eq:general_expression_for_Fisher}) also
holds in the case $q=1.$ The explicit expression of the generalized
Fisher information in the case of the generalized Gaussian is therefore
\begin{align}
I_{\beta,q}[G_{\gamma}] & =\left(\alpha\right)^{\beta}\left(\frac{2}{\alpha}n\,\omega_{n}\right)^{\beta(1-q)}\left|(q-1)\right|^{-n\frac{\beta}{\alpha}(1-q)-1}\gamma^{\frac{\beta}{\alpha}\left(n\left(q-1\right)+1\right)}\nonumber \\
 & \times\begin{cases}
\frac{B\left(1+\frac{n}{\alpha},-\frac{q}{q-1}-\frac{n}{\alpha}\right)}{\left[B\left(\frac{n}{\alpha},-\frac{1}{q-1}-\frac{1}{\alpha}\right)\right]^{\beta(q-1)+1}} & \text{for }\text{max}\left\{ 1-\alpha,\frac{n}{n+\alpha}\right\} <q<1\\
\frac{B\left(1+\frac{n}{\alpha},\frac{q}{q-1}\right)}{\left[B\left(\frac{n}{\alpha},\frac{q}{q-1}\right)\right]^{\beta(q-1)+1}} & \text{for }q>1,
\end{cases}\label{eq:ExpressionFisherGen}
\end{align}
and, $\text{for }q=1$, 
\begin{equation}
I_{\beta,q}[G_{\gamma}]=\left(\alpha\right)^{\beta}\left(\frac{2}{\alpha}n\,\omega_{n}\right)^{\beta(1-q)}\!\gamma^{\frac{\beta}{\alpha}\left(n\left(q-1\right)+1\right)}\frac{n}{\alpha}\,\Gamma\!\left(\frac{n}{\alpha}\right)^{\beta(1-q)}.\label{eq:ExpressionFisherGenb}
\end{equation}
Finally, let us note that we have the following simple scaling identities:
\begin{equation}
\begin{cases}
M_{q}[G_{\gamma}]=\gamma^{\frac{n}{\alpha}(q-1)}M_{q}(G),\\
I_{\beta,q}[G_{\gamma}]=\gamma^{\frac{\beta}{\alpha}\left(n\left(q-1\right)+1\right)}I_{\beta,q}[G],\\
m_{\alpha}[G_{\gamma}]=\gamma^{-1}m_{\alpha}[G].
\end{cases}\label{eq:ScaleEqualities}
\end{equation}


\end{document}